\documentclass[12pt]{extarticle}

\newcommand{\remove}[1]{}

\def\UseBibLatex{1}

\makeatletter
\def\input@path{{styles/}}
\makeatother

\providecommand{\BibLatexMode}[1]{}
\providecommand{\BibTexMode}[1]{}

\renewcommand{\BibLatexMode}[1]{#1}
\renewcommand{\BibTexMode}[1]{}

\ifx\UseBibLatex\undefined%
  \renewcommand{\BibLatexMode}[1]{}
  \renewcommand{\BibTexMode}[1]{#1}
\fi

\BibLatexMode{%
   \usepackage[bibencoding=utf8,style=alphabetic,backend=biber]{biblatex}%
   \usepackage{sariel_biblatex}%
}

\usepackage{amsmath}%
\usepackage{amssymb}%
\usepackage[table]{xcolor}%
\usepackage{microtype}
\usepackage{xurl}

\usepackage[amsmath,thmmarks]{ntheorem}%

\usepackage{titlesec}%
\usepackage{xcolor}%
\usepackage{mleftright}%
\usepackage{xspace}%
\usepackage{graphicx}
\usepackage{hyperref}%
\usepackage[inline]{enumitem}
\usepackage{hyperref}%
\usepackage[ocgcolorlinks]{ocgx2}

\hypersetup{%
      unicode,
      breaklinks,%
      colorlinks=true,%
      urlcolor=[rgb]{0.25,0.0,0.0},%
      linkcolor=[rgb]{0.5,0.0,0.0},%
      citecolor=[rgb]{0,0.2,0.445},%
      filecolor=[rgb]{0,0,0.4},
      anchorcolor=[rgb]={0.0,0.1,0.2}%
}

\titlelabel{\thetitle. }%

\theoremseparator{.}%

\theoremstyle{plain}%
\newtheorem{theorem}{Theorem}[section]

\newtheorem{lemma}[theorem]{Lemma}

\newtheorem{observation}[theorem]{Observation}

\theoremstyle{plain}%
\theoremheaderfont{\sf} \theorembodyfont{\upshape}%
\newtheorem*{remark:unnumbered}[theorem]{Remark}%
\newtheorem{remark}[theorem]{Remark}%

\newtheorem{defn}[theorem]{Definition}

\theoremheaderfont{\em}%
\theorembodyfont{\upshape}%
\theoremstyle{nonumberplain}%
\theoremseparator{}%
\theoremsymbol{\myqedsymbol}%
\newtheorem{proof}{Proof:}%

\providecommand{\emphind}[1]{}%
\renewcommand{\emphind}[1]{\emph{#1}\index{#1}}

\definecolor{blue25emph}{rgb}{0, 0, 11}

\providecommand{\emphic}[2]{}
\renewcommand{\emphic}[2]{\textcolor{blue25emph}{%
      \textbf{\emph{#1}}}\index{#2}}

\providecommand{\emphi}[1]{}%
\renewcommand{\emphi}[1]{\emphic{#1}{#1}}

\definecolor{almostblack}{rgb}{0, 0, 0.3}

\providecommand{\emphw}[1]{}%
\renewcommand{\emphw}[1]{{\textcolor{almostblack}{\emph{#1}}}}%

\providecommand{\emphOnly}[1]{}%
\renewcommand{\emphOnly}[1]{\emph{\textcolor{blue25emph}{\textbf{#1}}}}

\newcommand{\myqedsymbol}{\rule{2mm}{2mm}}
\newcommand{\SarielThanks}[1]{%
   \thanks{%
      School of Computing and Data Science; %
      University of Illinois; %
      201 N. Goodwin Avenue; %
      Urbana, IL, 61801, USA; %
      \href{mailto:spam@illinois.edu}{sariel@illinois.edu}; %
      \url{http://sarielhp.org/}.%
   #1%
   }%
}

\newcommand{\HLink}[2]{\hyperref[#2]{#1~\ref*{#2}}}
\newcommand{\HLinkSuffix}[3]{\hyperref[#2]{#1\ref*{#2}{#3}}}

\newcommand{\thmlab}[1]{{\label{theo:#1}}}
\newcommand{\thmref}[1]{\HLink{Theorem}{theo:#1}}

\newcommand{\lemlab}[1]{\label{lemma:#1}}
\newcommand{\lemref}[1]{\HLink{Lemma}{lemma:#1}}%

\newcommand{\remlab}[1]{\label{rem:#1}}
\newcommand{\remref}[1]{\HLink{Remark}{rem:#1}}%

\newcommand{\obslab}[1]{\label{observation:#1}}
\newcommand{\obsref}[1]{\HLink{Observation}{observation:#1}}

\newcommand{\seclab}[1]{\label{sec:#1}}
\newcommand{\secref}[1]{\HLink{Section}{sec:#1}}

\newcommand{\defrefY}[2]{\hyperref[def:#1]{#2}}

\providecommand{\eqlab}[1]{}%
\renewcommand{\eqlab}[1]{\label{equation:#1}}

\providecommand{\remove}[1]{}%
\newcommand{\Set}[2]{\left\{ #1 \;\middle\vert\; #2 \right\}}

\newcommand{\pth}[1]{\mleft(#1\mright)}%
\newcommand{\permut}[1]{\left\langle {#1} \right\rangle}

\newcommand{\ProbC}{{\mathbb{P}}}
\newcommand{\ExC}{{\mathbb{E}}}

\newcommand{\Prob}[1]{\ProbC\mleft[ #1 \mright]}
\newcommand{\Ex}[1]{\ExC\mleft[ #1 \mright]}

\newcommand{\ceil}[1]{\mleft\lceil {#1} \mright\rceil}

\newcommand{\brc}[1]{\left\{ {#1} \right\}}

\newcommand{\cardin}[1]{\left\lvert {#1} \right\rvert}%

\renewcommand{\th}{th\xspace}

\renewcommand{\Re}{\mathbb{R}}%

\newlist{compactenumA}{enumerate}{5}%
\setlist[compactenumA]{itemsep=-0.5ex,topsep=0.5ex,partopsep=1ex,parsep=1ex,%
   label=(\Alph*)}%

\newlist{compactenuma}{enumerate}{5}%
\setlist[compactenuma]{itemsep=-0.5ex,topsep=0.5ex,partopsep=1ex,parsep=1ex,%
   label=(\alph*)}%

\newlist{compactenumI}{enumerate}{5}%
\setlist[compactenumI]{itemsep=-0.5ex,topsep=0.5ex,partopsep=1ex,parsep=1ex,%
   label=(\Roman*)}%

\newlist{compactenumi}{enumerate}{5}%
\setlist[compactenumi]{itemsep=-0.5ex,topsep=0.5ex,partopsep=1ex,parsep=1ex,%
   label=(\roman*)}%

\newlist{compactitem}{itemize}{5}%
\setlist[compactitem]{itemsep=-0.5ex,topsep=0.5ex,partopsep=1ex,parsep=1ex,%
   label=\ensuremath{\bullet}}%

\newcommand{\etal}{\textit{et~al.}\xspace}

\numberwithin{figure}{section}%
\numberwithin{table}{section}%
\numberwithin{equation}{section}%

\usepackage{euscript}%
\usepackage{accents}%
\usepackage{bbm}%
\usepackage{wrapfig}

\newcommand{\eps}{{\varepsilon}}%
\newcommand{\poly}{\mathrm{poly}}

\newcommand{\normX}[1]{\left\| #1 \right\|}%
\newcommand{\dY}[2]{\left\| #1 - #2 \right\|}%
\newcommand{\Family}{\EuScript{F}}

\newcommand{\SphereChar}{\ensuremath{\mathbb{S}}}
\newcommand{\Sphere}[1]{{\SphereChar^{(#1)}}}

\newcommand{\DotProd}[2]{\permut{{#1},{#2}}}

\newcommand{\ball}{\mathbf{b}}%
\newcommand{\Center}{\mathbf{c}}

\newcommand{\Approx}[1]{\underaccent{~~#1}{\approx}}
\newcommand{\ApproxEps}[1]{{\underaccent{~~#1}{\approx}\,}}

\newcommand{\PntSet}{P}
\newcommand{\surface}{\gamma}
\newcommand{\surfaceB}{\eta}

\newcommand{\SurfSet}{\EuScript{G}}
\newcommand{\SurfSetB}{\EuScript{H}}
\newcommand{\SurfSetC}{\EuScript{I}}

\newcommand{\Coreset}{\EuScript{S}}
\newcommand{\CoresetB}{\EuScript{T}}
\newcommand{\CoresetC}{\EuScript{R}}

\newcommand{\face}{\mathrm{F}}

\newcommand{\DomainD}{\EuScript{D}_d}
\newcommand{\Domain}{\mathcal{D}}
\newcommand{\LinDim}{\mathbbmtt{d}}%

\newcommand{\LinMap}{{\mathsf{L}}}

\newcommand{\lowLim}{\mathsf{m}}
\newcommand{\hiLim}{\mathsf{M}}

\newcommand{\IntervalSet}{\mathfrak{K}}
\newcommand{\NumSet}{\mathbf{Z}}

\newcommand{\Prism}{\includegraphics[width=0.20cm]{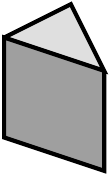}}

\newcommand{\seg}{\sigma}

\newcommand{\kapprox}{\widetilde{\kappa}}
\newcommand{\Polynomial}{\psi}

\newcommand{\dmY}[2]{\mathbf{d}({#2},{#1}) }
\newcommand{\LInfty}[1]{\left\|#1 \right\|_\infty}
\newcommand{\LOne}[1]{\left\|{#1} \right\|_1}
\newcommand{\LTwo}[1]{\left\|{#1}\right\|_2}
\newcommand{\LTwoOpt}{\ell_2}
\newcommand{\LOneOpt}{\ell_1}

\newcommand{\MdPrcX}[1]{{\nu_{\!\!\stackrel{\rule[-.005cm]{0cm}{0.05cm}}{#1}}}}
\newcommand{\MdOpt}{\nu_{\mathrm{opt}}}
\newcommand{\MdOptX}[1]{\nu_{\mathrm{opt}}(#1)}
\newcommand{\Decomp}{\Phi}
\newcommand{\num}{{\overline{\mathbf{z}}}}

\newcommand{\totalWeight}{W}

\newcommand{\SqPrcChar}{\mu}
\newcommand{\SqPrcX}[1]{{\SqPrcChar_{\!\stackrel{
            \rule[-.005cm]{0cm}{0.05cm}}{#1}}}}
\newcommand{\MdPrcXSQ}[1]{\SqPrcX{#1}}
\newcommand{\SqOptX}[1]{\SqPrcChar_{\mathrm{opt}}(#1)}

\newcommand{\popt}{p_{\mathrm{opt}}}

\newcommand{\VDist}{\mathsf{d}_|}
\newcommand{\Weight}[1]{{w_{#1}}}

\providecommand{\Copt}{C_{\mathrm{opt}}}

\providecommand{\Matousek}{Matou{\v s}ek\xspace}
\providecommand{\Err}{\EuScript{E}}
\newcommand{\Interval}{\mathcal{I}}

\newcommand{\Arr}{\mathop{\mathrm{\EuScript{A}}}}

\newcommand{\FuncSet}{\ensuremath{\mathcal{F}}}
\newcommand{\LVL}{\mathbf{L}}
\newcommand{\TLVL}{\mathbf{U}}

\newcommand{\RSample}{\Psi}
\newcommand{\RSampleB}{\Upsilon}
\newcommand{\EXT}[3]{\!\sideset{}{^{#2}_{#3}}{\mathop{#1|}}}

\newcommand{\pnt}{p}
\newcommand{\pntA}{q}
\newcommand{\probA}{\zeta}

\newcommand{\vecA}{\vec{v}}

\newcommand{\Cyl}{\mathfrak{C}}

\newcommand{\Flat}{\mathfrak{f}}
\newcommand{\FlatSet}{\mathfrak{F}}

\providecommand{\TPDF}[2]{\texorpdfstring{#1}{#2}}

\usepackage{qed_duck}

\bibliography{l1_fitting}

\begin{document}

\title{How to Get Close to the Median Shape%
   \thanks{%
      Alternative titles for this paper include: ``How to stay connected with your inner circle'' and ``How to compute one circle to rule them all''. The paper appeared in SoCG 2006 \cite{h-htgcm-06-conf}, and in a journal version in \cite{h-htgcm-07}. Th paper was slightly edited. It was uploaded to the arxiv at this late date to ensure its availability.%
   }%
}%

\author{Sariel Har-Peled\SarielThanks{Work on this paper was partially
      supported by an NSF CAREER award CCR-0132901.
      }
      }

\date{May 6, 2016\footnote{Paper Re\LaTeX{}ed on \today.}}

\maketitle

\begin{abstract}
    In this paper, we study the problem of $L_1$-fitting a shape to a set of $n$ points in $\Re^d$ (where $d$ is a fixed constant), where the target is to minimize the sum of distances of the points to the shape, or the sum of squared distances.  We present a general technique for computing a $(1 + \eps ) $-approximation for such a problem, with running time $O(n + \poly( \log n, 1/\eps))$, where $\poly(\log n, 1/\eps)$ is a polynomial of constant degree of $\log n$ and $1/\eps$ (the power of the polynomial is a function of $d$).  The new algorithm runs in linear time for a fixed $\eps>0$, and is the first subquadratic algorithm for this problem.

    Applications of the algorithm include best fitting either a circle, a sphere, or a cylinder to a set of points when minimizing the sum of distances (or squared distances) to the respective shape.
\end{abstract}

\section{Introduction}

Consider the problem of fitting a parameterized shape to given data.  This is a natural problem that arises in statistics, learning, data mining, and many other fields. What measure used for the quality of fitting has a considerable impact on the hardness of the problem of finding the best fitting shape.  As a concrete example, let $P$ be a set of $n$ points in $\Re^d$. A typical criterion for measuring how well a shape $\gamma$ fits $P$, denoted as $\mu(P,\gamma)$, is the maximum distance between a point of $P$ and its nearest point on $\gamma$, i.e., $\mu(P,\gamma) = \max_{p \in P} d(p,\gamma)$, where $d(p,\gamma) = \min_{q \in \gamma} dY{p}{q}$.  The extent measure of $P$ is $\mu(P) = \min_{\gamma \in \Family} \mu(P,\gamma)$, where $\Family$ is a family of shapes (such as points, lines, hyperplanes, spheres, etc.).  For example, the problem of finding the minimum radius sphere (resp.\ cylinder) enclosing $P$ is the same as finding the point (resp.\ line) that fits $P$ best, and the problem of finding the smallest width slab (resp.\ spherical shell, cylindrical shell) is the same as finding the hyperplane (resp.\ sphere, cylinder) that fits $P$ best.

A natural way of encoding the fitting information for a given shape $\gamma$ for the points of $P$, is by creating a point $\dmY{\gamma}{P} \in \Re^n$, where the $i$\th coordinate is the distance of the $i$th point of $P$ from $\gamma$.  Thus, the shape fitting problem mentioned above (of minimizing the distance to the furthest point to the shape) is to find the shape $\gamma$ that realizes $\min_{\gamma \in \Family} \LInfty{\dmY{\gamma}{P}}$. We will refer to this as the \emph{$L_\infty$-shape fitting problem}.

The exact algorithms for best shape fitting are generally expensive, e.g., the best known algorithms for computing the smallest volume bounding box containing $P$ in $\Re^3$ require $O(n^3)$ time \cite{o-fmeb-85}.  Consequently, attention has shifted to developing approximation algorithms \cite{bh-eamvb-01,zs-amves-02}. A general approximation technique was recently developed for such problems by Agarwal \etal \cite{ahv-aemp-04}.  This technique implies among other things that one can approximate the circle that best fit a set of points in the plane in $O(n + 1/\eps^{O(1)})$ time, where the fitting measure is the maximum distance of the point to the circle (in fact, this special case was handled before by Agarwal \etal{} \cite{aahs-aeamw-00} and by Chan \cite{c-adwse-02}).

The main problem with the $L_\infty$-fitting is its sensitivity to noise and outliers. There are two natural remedies.

The first is to change the target function to be less sensitive to outliers. For example, instead of considering the maximum distance, one can consider the average distance. This is the \emph{$L_1$-fitting problem}, and here we would like to compute the shape realizing $\displaystyle \LOneOpt(\Family, P) = \min_{\gamma \in \Family} \LOne{\dmY{\gamma}{P}} = \min_{\gamma \in \Family} \sum_{p \in P} d(p,\gamma)$. Similarly, in the \emph{$L_2$-fitting problem}, one would like to minimize the average squared distances of the points to the shape; namely, $\displaystyle \LTwoOpt(\Family, P) = \min_{\gamma \in \Family} \LTwo{\dmY{\gamma}{P}}^2 = \min_{\gamma \in \Family} \sum_{p \in P} \pth{d(p,\gamma)}^2$.  The $L_2$ fitting problem in the case of a single linear subspace is well understood and is computed via singular value decomposition (SVD).  Fast approximation algorithms are known for this problem; see \cite{fkv-fmcaf-04, drvw-mapcv-06} and references therein. As for the $L_1$-fitting of a linear subspace, this problem can be solved using linear programming techniques, in polynomial time in high dimensions, and linear time in constant dimension \cite{ykii-avo-88}.  Recently, Clarkson gave a faster approximation algorithm for this problem \cite{c-ssalo-05} which works via sampling.

The problem seems to be harder once the shape we consider is not a linear subspace. There is considerable work on nonlinear regressions \cite{sw-nr-89} (i.e., extension of the $L_2$ least squares technique) for various shapes. Still, there seems to be no efficient guaranteed approximation algorithm even for the ``easy'' problem of $L_1$-fitting a circle to the data. The hardness seems to arise from the target function being a sum of terms, each term being an absolute value of a difference of a square root of a polynomial and a radius (see \secref{problem_definition}). In fact, this is an extension of the Fermat-Weber problem, and it seems doubtful that an efficient exact solution would exist for such a problem.

The second approach is to specify a number $k$ of outliers in advance and find the best shape $L_\infty$-fitting all but $k$ of the input points. Har-Peled and Wang showed that there is a coreset for this problem \cite{hw-sfo-04}, and as such it can be solved in $O(n + \poly(k, \log n, 1/\eps))$ time, for a large family of shapes. The work of Har-Peled and Wang was motivated by the aforementioned problem of $L_1$-fitting a circle to a set of points. (The results of Har-Peled and Wang were recently improved by Agarwal \etal{} \cite{ahy-rsfpg-06}, but since the improvement is not significant for our purposes, we will stick with the older reference.)

\medskip

\paragraph{Our Results.}
In this paper, we describe a general technique for computing a $(1+\eps)$-approximate solution to the $L_1$ and $L_2$-fitting problems, for a family of shapes which is well behaved (roughly speaking, those are all the shapes that the technique of Agarwal \etal{} \cite{ahv-aemp-04} can handle). Our algorithm achieves a running time of $O(n + \poly( \log n, 1/\eps) )$.  As such, this work can be viewed as the counterpart to Agarwal \etal{} \cite{ahv-aemp-04} work on the approximate $L_\infty$-fitting problem. This is the first linear-time algorithm for this problem.

The only previous algorithm directly relevant for this result, we are aware of, is due to Har-Peled and Koltun \cite{hk-so-05} that, in $O(n^2 \eps^{-2} \log^2 n)$ time, approximates the best circle $L_1$-fitting a set of points in the plane.

\paragraph{Comment on running time.}
The running time of our algorithms is $O(n + \poly(\log n, 1/\eps)) = O(n + \poly(1/\eps))$.  However, throughout the paper, we use the former (and more explicit) bound to emphasize that the running time of the second stage of our algorithms depends on $n$, unlike other geometric approximation algorithms.

\medskip

\paragraph{Paper organization.}
In \secref{preliminaries} we introduce some necessary preliminaries.  In \secref{problem_statement} the problem is stated formally.  In \secref{chunking}, we provide a (somewhat bizarre) solution for the one-point $L_1$-fitting problem in one dimension (i.e., the one-median problem in one dimension).  In \secref{reduction}, we show how the problem size can be dramatically reduced.  In \secref{slow}, a slow approximation algorithm is described for the problem (similar in nature to the algorithm of \cite{hk-so-05}).  In \secref{main}, we state our main result and some applications. Conclusions are provided in \secref{conclusions}.

\section{Preliminaries}
\seclab{preliminaries}

Throughout the paper, we refer to the $x_d$-parallel direction in $\Re^d$ as \emph{vertical}. Given a point $x = (x_1, \ldots, x_{d-1})$ in $\Re^{d-1}$, let $(x, x_d)$ denote the point $(x_1, \ldots, x_{d-1}, x_d)$ in $\Re^d$.  Each point $x \in \Re^{d}$ is also a vector in $\Re^d$. Given a geometric object $A$, $A + x$ represents the object obtained by translating each point in $A$ by $x$.

A \emph{surface} is a subset of $\Re^d$ that intersects any vertical line in a single point. A \emph{surface patch} is a portion of a surface such that its vertical projection into $\Re^{d-1}$ is a semi-algebraic set of constant complexity, usually a simplex.  Let $A$ and $B$ be either a point, a hyperplane, or a surface in $\Re^d$. We say that $A$ lies \emph{above} (resp. \emph{below}) $B$, denoted by $A \succeq B$ (resp. $A \preceq B$), if for any vertical line $\ell$ intersecting both $A$ and $B$, we have that $x_d \geq y_d$ (resp. $x_d \le y_d$), where $(x_1, \ldots, x_{d-1}, x_d) = A \cap \ell$ and $(x_1, \ldots, x_{d-1}, y_d) = B \cap \ell$.  (In particular, if both $A$ and $B$ are hyperplanes, then $A \succeq B$ implies that $A$ and $B$ are parallel hyperplanes.)

Two non-negative numbers $x$ and $y$ are \emph{$(1 \pm \eps)$-approximation} of each other if $(1-\eps)x \leq y \leq (1+\eps) x$ and $(1-\eps)y \leq x \leq (1+\eps) y$. We denote this fact by $x \ApproxEps{\eps} y$.  Two non-negative functions $f(\cdot)$ and $g(\cdot)$ (defined over the same domain) are \emph{$(1 \pm \eps)$-approximation} of each other, denoted by $f \ApproxEps{\eps} g$, if $f(x) \ApproxEps{\eps} g(x)$, for all $x$.

\begin{observation}
    \obslab{approximate_transitive}%
    Let $x$ and $y$ be two positive numbers and $\eps < 1/4$. We have:
    \begin{compactenumi}
        \smallskip%
        \item If $x \ApproxEps{\eps} y$ and $y \ApproxEps{\eps} z$ then $x \Approx{3\eps} z$.
        \smallskip%
        \item If $\cardin{ x -y } \leq \eps x$ then $x \Approx{2\eps} y$.
        \smallskip%
        If ${ x } \leq (1+\eps) y$ and $y \leq (1+\eps) x$ then $x \ApproxEps{\eps} y$.
    \end{compactenumi}
\end{observation}

\subsection{Problem Statement}
\seclab{problem_statement}

\subsubsection{The Circle Fitting Case}
\seclab{circle_fitting}

To motivate our exposition, we will first consider the problem of $L_1$-fitting a circle to a set of points in the plane.

Let $\PntSet = \brc{p_1, \ldots, p_n}$ be a set of $n$ points in the plane, and consider the price $\MdPrcX{\PntSet} (C )$ of $L_1$-fitting the circle $C$ to $\PntSet$.  Formally, for a point $p_i \in \PntSet$ let $f_i(C) = \cardin{ \bigl. dY{p_i}{c} - r}$, where $c$ is the center of $C$, and $r$ is the radius of $C$.  Thus, the overall price, for a circle $C$ centered at $(x,y)$ with radius $r$, is
\begin{align*}
  \MdPrcX{\PntSet}(C)
  &=%
  \MdPrcX{\PntSet}(x,y,r) = \sum_{i=1}^n
    f_i(C)
    \\&%
  = \sum_{i=1}^n \cardin{ \bigl. dY{p_i}{c} - r} =
  \sum_{i=1}^n \cardin{ \bigl. \sqrt{ \pth{x_i - x}^2 + \pth{y_i
  - y}^2 } - r},
\end{align*}
where $p_i = (x_i,y_i)$, for $i=1,\ldots, n$.  We are looking for the circle $C$ minimizing $\MdPrcX{\PntSet}(C)$. This is the circle that best fits the point set under the $L_1$ metric. Let $\MdOpt(\PntSet)$ denote the price of the optimal circle $\Copt$.

Geometrically, each function $f_i$ induces a surface
\begin{equation*}
    \surface_i = \Set{ (x_p, y_p, \dY{p}{p_i}) }{ p \in \Re^2}
\end{equation*}
in 3D, which is a cone. Here, a circle is encoded by a 3D point
\begin{math}
    C= (x,y,r).
\end{math}
The value of $f_i(C)$ is the vertical distance between the point $C$ and surface $\surface_i$.  Thus, we have a set $\SurfSet$ of $n$ surfaces in 3D, and we are interested in finding the point that minimizes the sum of vertical distances of this point to the $n$ surfaces.

\subsubsection{The General Problem}
\seclab{problem_definition}

Formally, for a weighted set of surfaces $\SurfSet$ in $\Re^d$ and $p$ any point in $\Re^d$ let
\[
\MdPrcX{\SurfSet}(p) = \sum_{\surface \in
   \SurfSet} \Weight{\surface} \cdot \VDist(p, \surface )
\]
denote the \emph{$L_1$ distance of $p$ from $\SurfSet$}, where $ \VDist(p, \surface )$ is the vertical distance between $p$ and the surface $\surface$ and $\Weight{\surface}$ is the weight associated with $\surface$.  Throughout our discussion, weights are positive integer numbers.  If $\SurfSet$ is unweighted then any surface $\surface \in \SurfSet$ is assigned weight $\Weight{\surface}=1$. We would be interested in finding the point that minimizes $\MdPrcX{\SurfSet}(p)$ when $p$ is restricted to a domain $\DomainD$, which is a semi-algebraic set of constant complexity in $\Re^d$. This is the \emph{$L_1$-fitting} problem. The \emph{$L_2$-fitting problem} is computing the point $p \in \DomainD$ realizing the minimum of $\MdPrcXSQ{\SurfSet}(p) = \sum_{\surface \in \SurfSet} \Weight{\surface} \cdot \pth{\VDist(p, \surface )}^2$.

It would be sometime conceptually easier (e.g., see \secref{cylinder}) to think about the problem algebraically, where the $i$th surface $\surface_i$ is an image of a (non-negative) $(d-1)$-dimensional function $f_i(x_1, \ldots, x_{d-1})$ $= \sqrt{ p_i( x_1,\ldots, x_{d-1})}$, where $p_i(\cdot)$ is a constant degree polynomial, for $i=1,\ldots, n$.  We are interested in approximating one of the following quantities:
\begin{center}
\begin{tabular}{cc}
  (i)
  &
    $\displaystyle \min_{(x_1,\ldots, x_{d-1}) \in \Domain}
    \sum_{i=1}^n \Weight{i} \cdot f_i(x_1,\ldots,
    x_{d-1})$,\\
  (ii)
  & $\displaystyle
    \MdOptX{\SurfSet}
    = \min_{x  \in \DomainD} \MdPrcX{\SurfSet}(x) =
    \min_{(x_1,\ldots, x_{d}) \in \DomainD} \sum_i
    \Weight{i} \cdot \cardin{f_i(x_1,\ldots, x_{d-1}) - x_d}$,
  \\
  or (iii)
  &
    $\displaystyle
    \SqOptX{\SurfSet}
    =
    \min_{x  \in \DomainD} \SqPrcX{\SurfSet}(x)
    =%
    \!\!\!\!\! \min_{(x_1,\ldots, x_{d}) \in \DomainD} \sum_i
    \Weight{i} \cdot \pth{f_i(x_1,\ldots, x_{d-1}) - x_d}^2$,
\end{tabular}
\end{center}
where $\Domain \subseteq \Re^{d-1}$ and $\DomainD \subseteq \Re^d`$ are semi-algebraic sets of constant complexity, and the weights $\Weight{1},\ldots, \Weight{n}$ are positive integers.  Note that (i) is a special case of (ii), by setting $\DomainD = \Domain \times \brc{0}$.

To simplify the exposition, we will assume that $\DomainD=\Re^d$. It is easy to verify that our algorithm also works for the more general case with a few minor modifications.

\paragraph{The linearization dimension.}
In the following, a significant parameter in the exposition is the \emph{linearization dimension} $\LinDim$, which is the target dimension we need to map the polynomials $p_1,\ldots, p_n$ so that they all become linear functions. For example, if the polynomials are of the form $\Polynomial_i(x,y,z) = x^2 + y^2 + z^2 + a_ix+b_iy+ c_iz$, for $i=1,\ldots, n$, then they can be linearized by a mapping $\LinMap(x,y,z) = (x^2 +y^2 + z^2, x, y, z)$, such that $h_i(x,y,z,w) = w + a_ix+b_iy+ c_iz$ is a linear function and $\Polynomial_i(x,y,z) = h_i(\LinMap(x,y,z))$. Thus, in this specific example, the linearization dimension is $4$. The linearization dimension is always bounded by the number of different monomials appearing in the polynomials ${p_1,\ldots, p_n}$. Agarwal and \Matousek{}~\cite{am-rsss-94} described an algorithm that computes a linearization of the smallest dimension for a family of such polynomials.

\section{Approximate \TPDF{$L_1$}{L1}-Fitting in One Dimension}
\seclab{chunking}

In this section, we consider the one-dimensional problem of approximating the distance function of a point $z$ to a set of points $\NumSet = \permut{ z_1, z_2, \ldots, z_n}$, where $z_1 \leq z_2 \leq \ldots \leq z_n$. Formally, we want to approximate the function $\MdPrcX{\NumSet}(\num) = \sum_{z_i \in \NumSet} \cardin{z_i-\num}$.  This function is the one-median function for $\NumSet$ on the real line. This corresponds to a vertical line in $\Re^d$, where each $z_i$ represents the intersection of the vertical line with the surface $\surface_i$.  The one-dimensional problem is well understood, and there exists a coreset for it; see \cite{hm-ckmkm-04,hk-sckmk-07}.  Unfortunately, it is unclear how to extend these constructions to the higher-dimensional case; specifically, how to perform the operations required in a global fashion on the surfaces so that the construction would hold for all vertical lines. See \remref{hardness} below for more details on this ``hardness''.  Thus, we present here an alternative construction.

\begin{defn}
    For a set of weighted surfaces $\SurfSet$ in $\Re^d$, a weighted subset $\Coreset \subseteq \SurfSet$ is an \emph{$\eps$-coreset} for $\SurfSet$ if for any point $p \in \Re^d$ we have $\MdPrcX{\SurfSet}(p) \ApproxEps{\eps} \MdPrcX{\Coreset}(p)$.
\end{defn}

For the sake of simplicity of exposition, in the following we assume that $\SurfSet$ is unweighted. The weighted case can be handled in a similar fashion.

The first step is to partition the points. Formally, we partition $\NumSet$ symmetrically into subsets, such that the sizes of the subsets increase as one comes toward the middle of the set.  Formally, the set $L_i = \brc{ z_i}$ contains the $i$th point on the line, for $i=1,\ldots, \lowLim$, where $\lowLim \geq 10/\eps$ is a parameter to be determined shortly.  Similarly, $R_i = \brc{z_{n-i+1}}$, for $i=1, \ldots, \lowLim$. Set $\alpha_\lowLim = \lowLim$, and let $\alpha_{i+1} = \min \pth{\ceil{(1+\eps/10)\alpha_i}, n/2}$, for $i=\lowLim,\ldots,\hiLim$, where $\alpha_\hiLim$ is the first number in this sequence equal to $n/2$. Now, let $L_i = \brc{z_{\alpha_{i-1}+1}, \ldots, z_{\alpha_{i}} }$ and $R_i = \brc{z_{n- \alpha_{i-1}}, \ldots, z_{n - \alpha_{i}+1} }$, for $i=\lowLim+1,\ldots, \hiLim$.  We will refer to a set $L_i$ or $R_i$ as a \emph{chunk}. Consider the partition of $\NumSet$ formed by the chunks $L_1, L_2, \ldots, L_\hiLim,R_\hiLim, \ldots, R_2, R_1$.  This is a partition of $\NumSet$ into ``exponential sets''.  The first/last $\lowLim$ sets on the boundary are singletons, and all the other sets grow exponentially in cardinality, till they cover the whole set $\NumSet$.

Next, we pick arbitrary points $l_i \in L_i$ and $r_i \in R_i$ and assign them weight $\Weight{i} = \cardin{R_i} = \cardin{L_i}$, for $i=1,\ldots, \hiLim$. Let $\Coreset$ be the resulting weighted set of points. We claim that this is a coreset for the $1$-median function.

But before delving into this, we need the following technical lemma.
\begin{lemma}
    \lemlab{trivial}%
    Let $A$ be a set of $n$ real numbers, and let $\psi$ and $\num$ be any two real numbers. We have that $\cardin{\MdPrcX{A}(\num) - \cardin{A} \cdot \cardin{\psi - \num}} \leq \MdPrcX{A}(\psi)$.
\end{lemma}

\begin{proof}
    We have
    \begin{align*}
      \cardin{\bigl. \MdPrcX{A}(\num) - \cardin{A} \cdot
      \cardin{\psi - \num}}
      &=
        \cardin{\bigl. \sum_{p \in A}
        \cardin{p - \num} - \cardin{A} \cdot \cardin{\psi - \num}}
      \\&
      \leq%
      \sum_{p \in A} \cardin{\bigl. \cardin{\num - p} - \cardin{\psi
      - \num} }
      \leq%
      \sum_{p \in A} \cardin{p - \psi} %
      =%
      \MdPrcX{A}(\psi),
    \end{align*}
    by the triangle inequality.
\end{proof}

\begin{lemma}
    \lemlab{1_d_coreset}%
    It holds
    \begin{math}
        \MdPrcX{\NumSet}(\num) \Approx{\eps/5} \MdPrcX{\Coreset}(\num),
    \end{math}
    for any $\num\in \Re$.
\end{lemma}

\begin{proof}
    We claim that
    \[
        \cardin{\MdPrcX{\NumSet}(\num) - \MdPrcX{\Coreset}(\num) } \leq (\eps/10) \MdPrcX{\NumSet}(\num),
    \]
    for all $\num \in \Re$. Indeed, let $\tau$ be a median point of $\NumSet$ and observe that $\MdPrcX{\NumSet}( \tau)$ is a global minimum of this function. We have that
    \begin{align*}
      \Err
      &=
        \cardin{\MdPrcX{\NumSet}(\num) - \MdPrcX{\Coreset}(\num) }
        \leq%
        \sum_{i=1}^{\hiLim} \cardin{\MdPrcX{L_i}(\num) - \cardin{L_i}
           \cdot \cardin{l_i - \num}} +
        \sum_{i=1}^{\hiLim} \cardin{\MdPrcX{R_i}(\num) - \cardin{R_i}
           \cdot \cardin{r_i - \num}}\\
        &=
        \sum_{i=\lowLim+1}^{\hiLim} \cardin{\MdPrcX{L_i}(\num) - \cardin{L_i}
           \cdot \cardin{l_i - \num}}
        +
        \sum_{i=\lowLim+1}^{\hiLim} \cardin{\MdPrcX{R_i}(\num) - \cardin{R_i}
           \cdot \cardin{r_i - \num}}\\
        &\leq
        \sum_{i=\lowLim+1}^{\hiLim} \MdPrcX{L_i}(l_i) +
        \sum_{i=\lowLim+1}^{\hiLim} \MdPrcX{R_i}(r_i),
    \end{align*}
    by \lemref{trivial}.

    Observe that by construction $\cardin{R_i} \leq (\eps/10) \cardin{R_1 \cup \ldots \cup R_{i-1}}$, for $i > \lowLim$. We claim that $\sum_{i=\lowLim+1}^{\hiLim} \MdPrcX{L_i}(l_i) + \sum_{i=\lowLim+1}^{\hiLim} \MdPrcX{R_i}(r_i) \leq (\eps/10)\MdPrcX{\NumSet}(\tau)$. To see this, for each point of $z_i \in \NumSet$, let $I_i$ be the interval with $z_i$ in one endpoint and the median $\tau$ in the other endpoint.  The total length of those intervals is $\MdPrcX{\NumSet}(\tau)$. Let $\IntervalSet = \brc{I_1,\ldots, I_n}$.

    Consider the interval $\Interval_i =\Interval(R_i)$ which is the shortest interval containing the points of $R_i$, for $i=\lowLim+1, \ldots, \hiLim$. Clearly, we have $\MdPrcX{R_i}(r_i) \leq \cardin{R_i} \cdot \normX{\Interval_i}$.

    On the other hand, the number of intervals of $\IntervalSet$ completely covering $\Interval_i$ is at least $(10/\eps)\cardin{R_i}$, for $i=\lowLim+1,\ldots, \hiLim$. As such, we can charge the total length of $\MdPrcX{R_i}(r_i)$ to the portions of those intervals of $\IntervalSet$ covering $\Interval_i$.  Thus, every unit of length of the intervals of $\IntervalSet$ gets charged at most $\eps/10$ units.

    This implies that the error $\Err \leq (\eps/10) \MdPrcX{\NumSet}(\tau) \leq (\eps/10) \MdPrcX{\NumSet}(\num)$, which establishes the lemma, by \obsref{approximate_transitive}.
\end{proof}

Next, we ``slightly'' perturb the points of the coreset $\Coreset$.  Formally, assume that we have points $l_1',\ldots, l_\hiLim', r_1', \ldots, r_\hiLim'$ such that $\cardin{l_i' - l_i} ,\cardin{r_i' - r_i} \leq (\eps/20)\cardin{l_i - r_i}$, for $i=1,\ldots, N$. Let $\CoresetC = \brc{l_1',\ldots, l_\hiLim', r_\hiLim', \ldots, r_1'}$ be the resulting weighted set.  We claim that $\CoresetC$ is still a good coreset.

\begin{lemma}
    \lemlab{snap}%
    It holds that $\MdPrcX{\NumSet}(\num) \ApproxEps{\eps} \MdPrcX{\CoresetC}(\num)$, for any $\num\in \Re$.  Namely, $\CoresetC$ is an $\eps$-coreset for $\NumSet$.
\end{lemma}

\begin{proof}
    By \lemref{1_d_coreset} and by the triangle inequality, we have
    \begin{eqnarray*}
        \MdPrcX{\Coreset}(\num) &=& \sum_i \pth{ \cardin{L_i}
           \cdot \cardin{l_i - \num } + \cardin{R_i}\cdot \cardin{ r_i -
              \num} }
        \geq
        \sum_i { \cardin{L_i}
           \cdot \cardin{l_i - r_i} },
    \end{eqnarray*}
    since for all $i$ we have $\cardin{L_i} = \cardin{R_i}$.  Also, by the triangle inequality
    \begin{eqnarray*}
      \cardin{\cardin{l_i - \num} - \cardin{l_i' - \num} \bigl.} \leq \cardin{l_i - l_i'}.
    \end{eqnarray*}
    Thus
    \begin{align*}
      \cardin{\MdPrcX{\Coreset}(\num) - \MdPrcX{\CoresetC}(\num)}
      &\leq%
        \sum_i \cardin{L_i} \cdot \cardin{l_i - l_i'} + \sum_i \cardin{R_i}
        \cardin{r_i - r_i'}
        \\&
        \leq%
        2
        \sum_i \cardin{L_i} \frac{\eps}{20} \cdot \cardin{l_i -
        r_i}
        \leq \frac{\eps}{10} \nu_\Coreset(\num).
    \end{align*}
    Thus $\CoresetC$ is an $\eps/5$-coreset of $\Coreset$, which is in turn an $\eps$-coreset for $\NumSet$, by \obsref{approximate_transitive}.
\end{proof}

\begin{remark}
    \remlab{hardness}%
    The advantage of the scheme used in \lemref{snap} over the constructions of \cite{hm-ckmkm-04,hk-sckmk-07} is that the new framework is more combinatorial and therefore it is more flexible.  In particular, the construction can be done in an oblivious way without knowing (even approximately) the optimal value of the $1$-median clustering.  This is in contrast to the previous constructions that are based on the partition of the line into intervals of prespecified length that depend on the value of the optimal solution. As such, they can not be easily extended to handle noise and approximation.  The flexibility of the new construction is demonstrated in the following section.
\end{remark}

\subsection{Variants}

Let $f: \Re^+ \rightarrow \Re^+$ be a monotone strictly increasing function (e.g., $f(x) =x^2$). Consider the function
\[
U_\NumSet(\num) = \sum_{x \in \NumSet} f\pth{ \bigl.\cardin{x -
      \num}}.
\]
We claim that the set $\Coreset$ constructed in \lemref{1_d_coreset} is also a coreset for $U_\NumSet(\cdot)$. Namely, $U_\NumSet(\num) \Approx{\eps/5} U_\Coreset(\num) = \sum_{x \in \Coreset} \Weight{x} f\pth{ \bigl. \cardin{x - \num}}$.  To this end, map each point $x$ of $\NumSet$ to a point of distance $f\pth{ \bigl. \cardin{x - \num}}$ from $\num$ (preserving the side of $\num$ on which the point $x$ lies), and let $g_\num:\NumSet \rightarrow \Re$ denote this mapping. Let the resulting set be $Q= f(\NumSet)$.  Clearly, $U_\NumSet(\num) = \MdPrcX{Q}(\num)$, and let $\CoresetB$ be the coreset constructed for $Q$ by \lemref{1_d_coreset}.  Observe that $\CoresetB = g_\num(\Coreset)$, since the construction of the coreset cares only about the ordering of the points, and the ordering is preserved when mapping between $\NumSet$ and $Q$. Thus, we have that $U_\NumSet(\num) = \MdPrcX{Q}(\num) \Approx{\eps/5} \MdPrcX{\CoresetB}(\num) = U_\Coreset(\num)$, as required.

This in particular implies that $\SqPrcX{\NumSet}(\num) \Approx{\eps/5} \SqPrcX{\Coreset}(\num)$, for any $\num \in \Re$, where $\SqPrcX{\NumSet}(\num) = \sum_{x \in \NumSet} \cardin{\num - x}^2$.  In this case, even the modified coreset $\CoresetC$ is still a coreset.

\begin{lemma}
    \lemlab{snap_square}%
    It holds that $\SqPrcX{\NumSet}(\num) \ApproxEps{\eps} \SqPrcX{\CoresetC}(\num)$, for any $\num \in \Re$.  Namely, $\CoresetC$ is an $\eps$-coreset of $\NumSet$ for the $\SqPrcChar(\cdot)$ function.
\end{lemma}

\begin{proof}
    Observe that, by the above discussion, $\SqPrcX{\NumSet}(\num) \Approx{\eps/5} \SqPrcX{\Coreset}(\num)$. On the other hand, fix $\num \in \Re$, and assume that $|l_i-\num| < |r_i-\num|$. This implies that $|r_i - \num| \geq |l_i - r_i|/2$, and we have
    \begin{align*}
      &\cardin{ l_i' - \num}^2 +
        \cardin{ r_i' - \num}^2
        \leq
        \pth{ \cardin{l_i' - l_i} + \cardin{ l_i - \num}}^2
        + \pth{ \cardin{r_i' - r_i} + \cardin{ r_i - \num}}^2
      \\
      &\qquad\leq%
        \pth{ (\eps/10) \cardin{r_i - \num} + \cardin{ l_i - \num}}^2
        + \pth{ (\eps/10)\cardin{r_i - \num} + \cardin{ r_i - \num}}^2
      \\
      &\qquad\leq%
        (\eps^2/100) \cardin{r_i - \num}^2
        + (\eps/5) \cardin{r_i - \num} \cardin{ l_i - \num}\\
      &\qquad\qquad +  \cardin{ l_i - \num}^2 +
        (1+\eps/10)^2 \cardin{ r_i - \num}^2
      \\&%
      \qquad\leq  (1+\eps/3)\pth{
      \cardin{ l_i - \num}^2 +
      \cardin{ r_i - \num}^2
      },
    \end{align*}
    since $\cardin{l_i' - l_i} ,\cardin{r_i' - r_i} \leq (\eps/20)\cardin{l_i - r_i}$. This implies that $\SqPrcX{\Coreset}(\num) \leq (1+\eps/3)\SqPrcX{\CoresetC}(\num)$.  By applying the same argument in the other direction, we have that $\SqPrcX{\Coreset}(\num) \Approx{\eps/3} \SqPrcX{\CoresetC}(\num)$, by \obsref{approximate_transitive} (iii). This in turn implies that $\SqPrcX{\CoresetC}(\num) \ApproxEps{\eps} \SqPrcX{\NumSet}(\num)$, as required.
\end{proof}

\section{The Reduction}
\seclab{reduction}

In this section, we show how to reduce the problem of approximating the $\MdPrcX{\SurfSet}(\cdot)$ function, for a set $\SurfSet$ of $n$ (unweighted) surfaces in $\Re^d$, to the problem of approximating the same function for a considerably smaller set of surface patches.

\secref{chunking} provides us with a general framework for how to get a small approximation. Indeed, pick any vertical line $\ell$, and consider its intersection points with the surfaces of $\SurfSet$.  Clearly, the function $\MdPrcX{\SurfSet}(\cdot)$ restricted to $\ell$ can be approximated using the construction of \secref{chunking}.  To this end, we need to pick levels in the way specified and assign them the appropriate weights. This would guarantee that the resulting function would approximate $\MdPrcX{\SurfSet}(\cdot)$ everywhere.

A significant difficulty in pursuing this direction is that the levels we pick have high descriptive complexity.  We circumnavigate this difficulty in two stages. In the first stage, we replace those levels with shallow levels by using random sampling. In the second stage, we approximate these shallow levels such that this introduces a small relative error.

\begin{defn}
    For a set $\SurfSet$ of $n$ surfaces in $\Re^d$, the \emph{level} of a point $x \in \Re^d$ in the arrangement $\Arr(\SurfSet)$ is the number of surfaces of $\SurfSet$ lying vertically below $x$.  For $k=0, \ldots, n-1$, let $\LVL_{\SurfSet, k}$ represent the surface which is the closure of all points on the surfaces of $\SurfSet$ whose level is $k$. We will refer to $\LVL_{\SurfSet,k}$ as the \emph{bottom $k$-level} or just the \emph{$k$-level} of $\SurfSet$.  We define the {\em top $k$-level} of $\SurfSet$ to be $\TLVL_{\SurfSet, k} = \LVL_{\SurfSet, n - k - 1}$, for $k=0,\ldots, n-1$.  Note that $\LVL_{\SurfSet, k}$ is a subset of the arrangement of $\SurfSet$.  For $x \in \Re^{d-1}$, we slightly abuse notation and define $\LVL_{\SurfSet, k}(x)$ to be the value $x_d$ such that $(x, x_d) \in \LVL_{\SurfSet, k}$.
\end{defn}

\begin{lemma}
    \lemlab{random_sample}%
    Let $\SurfSet$ be a set of $n$ surfaces in $\Re^d$, $0 < \delta < 1/4$, and let $k$ be a number between $0$ and $n/2$. Let $\probA = \min\pth{ c k^{-1}\delta^{-2} \log n, 1}$, and pick each surface of $\SurfSet$ into a random sample $\RSample$ with probability $\probA$, where $c$ is an appropriate constant. Then, with high probability, the $\kapprox$-level of $\Arr(\RSample)$ lies between the $(1-\delta)k$-level to the $(1+\delta)k$-level of $\Arr(\SurfSet)$, where $\kapprox = \probA k = O( \delta^{-2} \log n )$.

    In other words, we have $\LVL_{\SurfSet,(1-\delta)k} \preceq \LVL_{\RSample, \kapprox} \preceq \LVL_{\SurfSet, (1+\delta)k}$ and $\TLVL_{\SurfSet,(1+\delta)k} \preceq \TLVL_{\RSample, \kapprox} \preceq \TLVL_{\SurfSet, (1-\delta)k}$.
\end{lemma}

\begin{proof}
    The claim follows readily from the Chernoff inequality, and the proof is provided only for the sake of completeness.

    Consider a vertical line $\ell$ passing through a point $\pnt \in \Re^{d-1}$, and let $X_i$ be an indicator variable which is $1$ if the $i$th surface intersecting $\ell$ (from the bottom) was chosen for the random sample $\RSample$. Let $Y = \sum_{i=1}^{(1+\delta)k} X_i$ be the random variable which is the level of the point $\pth{\pnt, \LVL_{\SurfSet,(1+\delta)k}(\pnt)}$ in the resulting arrangement $\Arr(\RSample)$.

    Let $\mu = \Ex{Y} = \probA (1+\delta)k$. We have by the Chernoff inequality that
    \begin{align*}
      \Prob{ Y < \probA k }
      &\leq%
        \Prob{ Y < (1-\delta/2) \mu }\leq \exp
        \pth{ -\mu \delta^2 /8} = \exp\pth{ - \frac{(1+\delta)c}{8} \log n
        }
      \\&
      =%
      \frac{1}{n^{O(1)}},
    \end{align*}
    by choosing $c$ to be large enough. There are only $n^{O(1)}$ combinatorially different orderings of the surfaces of $\SurfSet$ along a vertical line. As such, we can make sure that, with high probability, the $\kapprox$ level in $\RSample$ (which is just a surface) lies below the $(1+\delta)k$ level of $\SurfSet$.

    A similar argument shows that, with high probability, the $(1-\delta)k$ level of $\SurfSet$ lies below the $\kapprox$ level of $\RSample$.
\end{proof}

\lemref{random_sample} suggests that instead of picking a specific level in a chunk of levels, as done in \secref{chunking}, we can instead pick a level, which is a shallow level of the appropriate random sample, and with high probability this level lies inside the allowable range.  The only problem is that even this shallow level might (and will) have unreasonable complexity. We rectify this by doing a direct approximation of the shallow levels.

\begin{defn}
    Let $\SurfSet$ be a set of surfaces in $\Re^d$.  The \emph{$(k,r)$-extent} $\EXT{\SurfSet}{k}{r} : \Re^{d-1} \rightarrow \Re$ is defined as the vertical distance between the bottom $r$-level and the top $k$-level of $\Arr(\SurfSet)$, i.e., for any $x \in \Re^{d-1}$, we have
    \[
    \EXT{\SurfSet}{k}{r} (x) = \TLVL_{\SurfSet, k}(x) -
    \LVL_{\SurfSet, r}(x).
    \]
\end{defn}

\begin{defn}[\cite{hw-sfo-04}]
    Let $\FuncSet$ be a set of non-negative functions defined over $\Re^{d-1}$.  A subset $\FuncSet' \subseteq \FuncSet$ is \emph{$(k,\eps)$-sensitive} if for any $r \leq k$ and $x \in \Re^{d-1}$, we have
    \[
        \LVL_{\FuncSet,r}(x) \leq \LVL_{\FuncSet',r}(x) \leq
        \LVL_{\FuncSet,r}(x) +
        \frac{\eps}{2}\EXT{\FuncSet}{k}{r}(x) ; ~~~~\text{and}
    \]
    \[
        \TLVL_{\FuncSet,r}(x) -
        \frac{\eps}{2}\EXT{\FuncSet}{r}{k}(x)
        \leq
        \TLVL_{\FuncSet',r}(x) \leq
        \TLVL_{\FuncSet,r}(x).
    \]
\end{defn}

We need the following result of Har-Peled and Wang \cite{hw-sfo-04}.  It states that for well well-behaved set of functions, one can find a small subset of the functions such that the vertical extent of the subset approximates the extent of the whole set. This holds only for ``shallow'' levels $\leq k$. In our application $k$ is going to be about $O(\eps^{-2} \log n)$. Here is the legalese:
\begin{theorem}[\cite{hw-sfo-04}]
    \thmlab{d_sensitive} %
    Let $\FuncSet = \brc{p_1^{1/2}, \ldots, p_n^{1/2}}$ be a family of $d$-variate functions defined over $\Re^d$, where $p_i$ is a $d$-variate polynomial, for $i=1,\ldots, n$.  Given $k$ and $0 < \eps < 1$, one can compute, in $O(n+ k/\eps^{2\LinDim})$ time, a subset $\FuncSet' \subseteq \FuncSet$, such that, with high probability, $\FuncSet'$ is $(k,\eps)$-sensitive for $\FuncSet$, and $\cardin{\FuncSet'} = O(k/\eps^{2\LinDim})$, where $\LinDim$ is the linearization dimension of the polynomials of $\FuncSet$.
\end{theorem}

Intuitively, \thmref{d_sensitive} states that shallow levels of depth at most $k$ have an approximation of size polynomial in $k$ and $1/\eps$, and matching bottom/top $k$ levels have their mutual distances preserved up to a small multiplicative factor.

\paragraph{The construction.}
We partition the levels of $\Arr(\SurfSet)$ into chunks, according to the algorithm of \secref{chunking}, setting $\lowLim = O( (\log n)/\eps^2 )$. The first top/bottom $\lowLim$ levels of $\Arr(\SurfSet)$ we approximate directly by computing a set $\Coreset_0$ which is $(\lowLim,\eps/20)$-sensitive for $\SurfSet$, using \thmref{d_sensitive}. Next, compute the $i$th bottom (resp., top) level of $\Arr(\Coreset_0)$, for $i=0,\ldots, \lowLim$, and let $\surface_i$ (resp., $\surfaceB_i$) denote those levels. We assign weight one to each such surface.

For every pair of chunks of levels $L_i$ and $R_i$ from \secref{chunking}, for $i =\lowLim+1, \ldots, \hiLim$, we compute an appropriate random sample $\RSample_i$. We remind the reader that $L_i$ spans the range of levels from $\alpha_{i-1}+1$ to $(1+\eps/10)\alpha_{i-1}$; see \secref{chunking}. As such, if we want to find a random level that falls inside this range, we need to set $\delta = \eps/40$ and $k= (1+\eps/20)\alpha_{i-1}$, and now apply \lemref{random_sample}, which results in a random set $\RSample_i$, such that level $l_i = O( \eps^{-2} \log n )$ of $\Arr(\RSample_i)$ lies between level $\alpha_{i-1}+1$ and $(1+\eps/10)\alpha_{i-1}$ of $\Arr(\SurfSet)$.  We now approximate the top $l_i$-level and bottom $l_i$-level of $\Arr(\RSample_i)$ by applying \thmref{d_sensitive} to $\RSample_i$. This results in a set $\Coreset_{i}$ of size $O(\poly(\log n, 1/\eps) )$ of surfaces, such that the extent of the top/bottom $l_i$ levels of $\Arr(\Coreset_{i-M})$, is an $(1 \pm \eps/40)$-approximation to the extent of the top/bottom $l_i$ levels in $\Arr(\RSample_i)$. We extract the bottom $l_i$ level and top $l_i$ level of $\Arr(\RSample_i)$. Let the two resulting surfaces be denoted by $\surface_i$ and $\surfaceB_i$, respectively, and assign them weight $\cardin{R_i}$, for $i=\lowLim+1,\ldots, \hiLim$.

Note that $\surface_i$ and $\surfaceB_i$ no longer have constant complexity, but their complexity is bounded by $O( \poly( \log n, 1/\eps) )$. Let $\SurfSetB = \brc{\surface_1, \surfaceB_1, \ldots, \surface_\hiLim, \surfaceB_\hiLim}$ be the resulting set of weighted surfaces, and observe that the complexity of the arrangement $\Arr(\SurfSetB)$ is $O( \poly( \log n, 1/\eps) )$. Furthermore, the analysis of \secref{chunking} implies that $\MdPrcX{\SurfSet}(\pnt) \ApproxEps{\eps} \MdPrcX{\SurfSetB}(\pnt)$, for any point $\pnt \in \Re^d$.

\paragraph{Implementation details.}
To get a linear running time, we need to carefully implement the above algorithm. First, observe that we computed $O(\eps^{-1} \log n )$ random samples $\RSample_{\lowLim+1}, \ldots, \RSample_{\hiLim}$.  Observe that if two random samples are generated by sampling every surface with probabilities that are similar (up to a factor of two), then we can use the same random sample. Thus, we need to generate random samples only for probabilities which are powers of two (implying that only $O( \log n)$ random samples are needed).  In particular, let $\RSampleB_i$ be a random sample generated by picking each surface of $\SurfSet$ with probability $1/2^{i}$.

To perform this sampling quickly, we generate the $(i+1)$th random sample by picking each surface of $\RSampleB_{i}$ into $\RSampleB_{i+1}$ with probability half (the sequence of random samples $\SurfSet = \RSampleB_0, \RSampleB_1,\ldots, \RSampleB_{O(\log n)}$ is sometimes referred to as a \emph{gradation}).  Namely, each $\RSampleB_i$ serves as a replacement for a sequence of random samples $\RSample_{\alpha},\ldots \RSample_{\beta}$ which were generated using similar probabilities, where $\alpha$ and $\beta$ are functions of $i$.

Next, we need to approximate the ``shallow'' levels of $\RSample_i$ up to level $\xi = O ( \max(l_{j_i},\ldots,$ $l_{j_{i+1}-1} ) )$ $= O( \eps^{-2} \log n )$. Again, we are performing the computation of the shallow levels for a batch of samples of $\RSample$ using a single sample of $\RSampleB$ (i.e., will approximate the top/bottom $O(\xi)$-levels of $\RSampleB_i$ and this would supply us with the surfaces approximating all the required levels in $\RSample_{\alpha},\ldots \RSample_{\beta}$).  Using \thmref{d_sensitive}, this takes $O( \cardin{\RSampleB_i} + \poly(\log n, 1/\eps ) )$ time. By the Chernoff inequality, with high probability, we have $\cardin{\RSampleB_{i}} = O( n/2^i)$. Thus, the overall running time, with high probability, is $\sum_i O( n/2^i + \poly(\log n, 1/\eps)) = O(n + \poly(1/\eps, \log n))$. Putting everything together, we have:

\begin{theorem}
    \thmlab{reduction}%
    Given a set of $n$ unweighted surfaces $\SurfSet$ in $\Re^d$ (i.e., as in \secref{problem_definition}), and a parameter $\eps$, one can compute a set $\SurfSetB$ of surface patches, such that each patch is a portion of a surface of $\SurfSet$ which is defined over a region in $\Re^{d-1}$ (such a region is a semi-algebraic set of constant descriptive complexity). The number of surface patches is $O( \poly(1/\eps, \log n) )$. Furthermore, $\MdPrcX{ \SurfSet }( \pnt ) \ApproxEps{ \eps } \MdPrcX{ \SurfSetB }(\pnt)$ and $\SqPrcX{\SurfSet}(\pnt) \ApproxEps{ \eps } \SqPrcX{\SurfSetB}( \pnt )$, for any point $\pnt \in \Re^d$. The algorithm takes $O(n + \poly(\log n, 1/\eps ) )$ time, and it succeeds with high probability.

    The total weight of surfaces intersecting any vertical line $\ell$
    is equal to $\cardin{\SurfSet}$.
\end{theorem}

The algorithm of \thmref{reduction} is a Monte-Carlo algorithm. In particular, it might fail with low probability. It is not clear if there is an efficient way to detect such a (low probability) failure.

\thmref{reduction} shows that given an instance of any of the problems defined in \secref{problem_definition}, we can quickly reduce the problem size to a small weighted set of surface patches. This, while beneficial, still leaves us with the mundane task of solving the problem on the reduced instance. Since we no longer have to care too much about efficiency, the problem becomes more manageable, and we tackle it in the next section.

\section{A Slow Approximation Algorithm}
\seclab{slow}

Let $\SurfSet$ be a set of $n$ weighted surface patches in $\Re^d$, such that any vertical line intersects surfaces with total weight $\totalWeight$.  In this section, we show how to solve any of the problems of \secref{problem_definition}.  A (roughly) quadratic time algorithm for the special case of a circle was given by Har-Peled and Koltun \cite{hk-so-05}, and the algorithm described here is somewhat similar to their algorithm. We demonstrate our algorithm for the case of approximating $\MdPrcX{\SurfSet}(\pnt) = \sum_{\surface \in \SurfSet} \Weight{\surface} \cdot \VDist(\pnt, \surface)$.

Let $\Decomp$ be the decomposition of $\Re^{d-1}$ into constant complexity cells, such that for each cell $\Delta \in \Decomp$, we have that any two vertical lines in $\Re^d$ intersecting $\Delta$ cross the same set of surface patches of $\SurfSet$.  Thus, $\Decomp$ induces a decomposition of $\Re^d$ into vertical prisms, such that we have to solve our problem inside each such prism.  The number of different prisms is $O\pth{n^{2\LinDim}}$, where $\LinDim$ is the linearization dimension of $\SurfSet$.  Now, we need to solve the problem inside each prism, for a weighted set of \emph{surfaces} (instead of surface patches).

So, consider a cell $\Delta \in \Decomp$ and let $\Prism$ denote the vertical prism that has $\Delta$ for a base.  Let $\SurfSetB$ be the resulting set of surfaces active in $\Prism$.  We compute, in $O(n)$ time, a vertical segment $\seg\subseteq \Prism$ that stabs all the surfaces of $\SurfSetB$, and its length is at most twice the length of the shortest vertical segment that intersects all the surfaces of $\SurfSetB$ inside $\Prism$. The algorithm of \cite{ahv-aemp-04} can be used here to this end.

The basic idea is to replace the ``unfriendly'' distance function $\VDist(\pnt, \surface)$, associated with $\surface \in \SurfSetB$, appearing in $\MdPrcX{\SurfSet}(\pnt)$ by its level-sets.  Namely, for each term in the summation of $\MdPrcX{\SurfSet}(\pnt) $, we will generate several level-sets, such that instead of computing $\VDist(\pnt, \surface)$, we will use the relevant level-set value.  Somewhat imprecisely, the level-set of $\VDist(\pnt,\surface)$ is a surface, and the value associated with the region between two consecutive level-sets will be the value $\VDist(\cdot, \surface)$ on the higher level-set. This process is somewhat similar to height contours used in drawing topographical maps. Since every level set is a surface, this induces an arrangement of surfaces.  For any point $\pnt \in \Re^d$, we can now compute $\MdPrcX{\SurfSet}(\pnt) $ by just figuring out in between what level-sets $\pnt$ lies. Therefore, evaluating $\MdPrcX{\SurfSet}(\pnt) $ is reduced to performing a point-location query in the arrangement of surfaces and returning the value associated with the face containing $\pnt$.

\begin{wrapfigure}{r}{3.2cm}
    \phantom{}\hfill%
    \includegraphics[width=3cm]{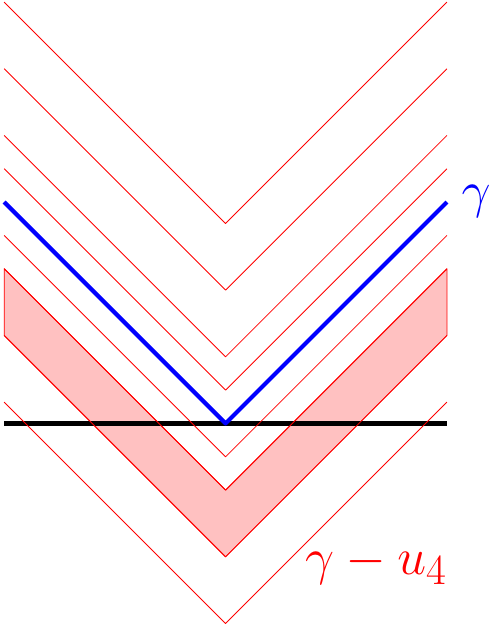}
\end{wrapfigure}
Note that $\normX{\seg}/2$ is a lower bound for the value of $\MdPrcX{\SurfSetB}(\cdot)$ in this prism and $\totalWeight \cdot \normX{\seg}$ is an upper bound on the value of $\MdOpt = \min_{\pnt \in \Prism} $ $\MdPrcX{\SurfSetB}(\pnt)$, where $\totalWeight = \Weight{\SurfSetB} \geq n$, where $\Weight{\SurfSetB}$ denotes the total weight of the surfaces of $\SurfSetB$.  As such, let $u = \eps \normX{\seg}/(10\totalWeight^2)$.  Next, let $u_i = i u$, for $i=1,\ldots, \lowLim = O(1/\eps)$. Let $u_i = (1+\eps/20) u_{i-1}$, for $i=\lowLim + 1,\ldots, \hiLim = O( \log_{1+\eps/20} \totalWeight) = O( \eps^{-1} \log \totalWeight )$. Note that $u_\hiLim > W^2 \normX{\seg} > 2 \MdOpt$.  Thus, if a point $\pnt \in \Prism$ is at a distance more than $u_\hiLim$ away from any surface of $\SurfSetB$, then its distance from this single surface is enough to ensure that $\MdPrcX{\SurfSetB}(\pnt) > 2\MdOpt$. In particular, for every surface $\surface \in \SurfSetB$, we create $O( \hiLim)$ copies of it as follows: $\surface - u_\hiLim, \surface - u_{\hiLim-1}, \ldots, \surface - u_1$, $\surface$, $\surface + u_1, \ldots, \surface + u_\hiLim$, where $\surface + x$ is the result of translating the surface $\surface$ up by distance $x$. Let $\SurfSetC(\surface)$ denote the resulting ``stack'' (i.e., set) of surfaces.

The surfaces of $\SurfSetC(\surface)$ partition the (domain of interest in the) prism into regions where the function $\VDist(\pnt, \surface )$ is the same up to a multiplicative factor of $(1+\eps/20)$. The only region that fails to comply with this condition is the region in between $\surface-u_1$ and $\surface + u_1$. Thus, if we approximate the value of $\VDist(\pnt, \surface )$ by the value of this function on the surface in this stack just above $\pnt$, we get an $(1 \pm \eps/20)$-approximation of this function, except for the region between $\surface - u_1$ and $\surface + u_1$.

In particular, let $\SurfSetC = \cup_{\surface \in \SurfSetB} \SurfSetC(\surface)$. Consider any point $\pnt \in \Prism$, and let $a_i = \VDist(\pnt, \surface_i)$, where $\surface_i \in \SurfSetB$, for $i=1,\ldots, n$, such that $\cardin{a_i} \leq u_\hiLim$. Also, let $b_i$ be the maximum of the values associated with the surfaces just above and below $\pnt$ in the stack $\SurfSetC(\surface_i)$. Thus, we have that
\begin{align*}
  \MdPrcX{\SurfSetB}(\pnt)
  &\leq%
    \overline{\MdPrcX{\SurfSetB}}(\pnt) =
    \sum_i b_i \leq \sum_i \pth{ u + \pth{1+\frac{\eps}{20}} a_i}
  \\&
  =%
    n
    \cdot \frac{\eps \normX{\seg}}{10\totalWeight^2} +
    \pth{1+\frac{\eps}{20}} \sum_i a_i%
  \\
  &\leq%
    \pth{1+\frac{\eps}{5}}
    \MdPrcX{\SurfSetB}(\pnt),
\end{align*}
since $\MdPrcX{\SurfSetB}(\pnt) \geq \normX{\seg}$.

Thus, let $\Arr$ be the arrangement $\Arr(\SurfSetC)$, and compute for every face $\face$ of $\Arr$ the value of $\overline{ \MdPrcX{ \SurfSetB } } (\pnt)$, where $\pnt \in \face$. By the above argument, the point $\popt$ realizing $\min_{x \in \Prism} \MdPrcX{\SurfSetB}(x)$ is approximated correctly by $\overline{\MdPrcX{\SurfSetB}}(\popt)$. As such, any point inside the face of $\Arr$ with the lowest associated value of $\overline{\MdPrcX{\SurfSetB}}$ is the required approximation.

It is easy to verify that the same algorithm with minor modifications would also enable us to approximate the minimum of the mean function $\SqPrcX{\SurfSet}(\cdot)$.  We conclude:
\begin{theorem}
    \thmlab{slow}%
    Let $\SurfSet$ be a set of $n$ weighted surface patches in $\Re^d$, with linearization dimension $\LinDim$, such that any vertical line intersects surfaces with total weight $\totalWeight$. Let $\eps > 0$ be a parameter. Then one can compute, in $O \pth{ n^{3\LinDim+ 1} \eps^{-\LinDim} \log^\LinDim \totalWeight }$ time, a point $x \in \Re^d$, such that $\MdPrcX{\SurfSet}(x) \leq (1+\eps)\MdOptX{\SurfSet}$, where $\MdOptX{\SurfSet} = \min_{x \in \Re^d} \MdPrcX{\SurfSet}(x)$.

    One can also compute, in the same time complexity, a point $y \in \Re^d$, such that $\SqPrcX{\SurfSet}(y) \leq (1+\eps)\SqOptX{\SurfSet}$, where $\SqOptX{\SurfSet} = \min_{x \in \Re^d} \SqPrcX{\SurfSet}(x)$.
\end{theorem}

\begin{proof}
    The correctness follows from the above discussion. As for the running time, there are $O(n^{2\LinDim})$ prisms. In each prism, we have at most $n$ surfaces, and every surface gets replicated $O( \eps^{-1} \log \totalWeight )$ times. The complexity of the arrangement inside each prism is $O \pth{ \pth{n \eps^{-1} \log \totalWeight }^{\LinDim+1} }$. A careful implementation would require time proportional to the complexity of all those arrangements, which is $O \pth{ n^{3\LinDim+ 1} \eps^{-\LinDim} \log^\LinDim \totalWeight }$, as claimed.
\end{proof}

\section{The Main Result and Some Applications}
\seclab{main}

By plugging \thmref{reduction} into \thmref{slow}, we get the main result of this paper:

\begin{theorem}
    \thmlab{main}%
    Given a set of $n$ unweighted surfaces $\SurfSet$ in $\Re^d$, as defined in \secref{problem_definition}, and a parameter $0 < \eps<1/4$, then one can compute, in $O \pth{ n + \poly( \log n, 1/\eps ) }$ time, a point $x \in \Re^d$, such that $\MdPrcX{\SurfSet}(x) \leq (1+\eps)\MdOptX{\SurfSet}$, where $\MdOptX{\SurfSet} = \min_{x \in \Re^d} \MdPrcX{\SurfSet}(x)$.

    One can also compute, in the same time complexity, a point $y \in \Re^d$, such that $\SqPrcX{\SurfSet}(y) \leq (1+\eps)\SqOptX{\SurfSet}$, where $\SqOptX{\SurfSet} = \min_{x \in \Re^d} \SqPrcX{\SurfSet}(x)$.

    The algorithm is randomized and succeeds with high probability.
\end{theorem}

\subsection{Applications}

The discussion \secref{circle_fitting} implies that we can readily apply \thmref{main} to the problem of $L_1$-fitting a circle to a set of points in the plane. Note that in fact the same reduction would work for the $L_2$-fitting problem, and for fitting a sphere to points in higher dimensions. We conclude:
\begin{theorem}[$L_1/L_2$-fitting to a circle/sphere]
    \thmlab{sphere_fitting}%
    Consider a set $P$ of $n$ points in $\Re^d$, and $\eps >0$ a parameter. One can $(1+\eps)$-approximate the sphere best fitting the points of $P$, where the price is the sum of Euclidean distances of the points of $P$ to the sphere. The running time of the algorithm is $O(n + \poly( \log n, 1/\eps ))$, and the algorithm succeeds with high probability.

    Similarly, this yields a $(1+\eps)$-approximation to the sphere minimizing the sum of square distances of the points to the sphere.
\end{theorem}

To our knowledge, \thmref{sphere_fitting} is the first subquadratic algorithm for this problem. A roughly quadratic time algorithm for the problem of $L_1$-fitting a circle to points in the plane was provided by Har-Peled and Koltun \cite{hk-so-05}.

\subsubsection{\TPDF{$L_1/L_2$}{L1/L2}-Fitting a cylinder to a point-set}
\seclab{cylinder}

Let $P = \brc{\pnt_1,\ldots, \pnt_n}$ be a set of $n$ points in $\Re^d$, $\ell$ be a line in $\Re^d$ parameterized by a point $\pntA \in \ell$, and a direction $\vecA$ on the unit sphere $\Sphere{n} \subseteq \Re^d$, and let $r$ be the radius of the cylinder having $\ell$ as its center. We denote by $\Cyl = \Cyl(\pntA,\vecA,r)$ the cylinder having $\displaystyle \ell = \cup_{t \in \Re} \pth{ \pntA + t \vecA \,}$ as its center. For a point $\pnt_i \in P$, we have that its distance from $\Cyl$ is
\[
    f_i(\pntA, \vecA, r) = \dmY{\pnt_i}{ \Cyl}%
    =%
    \cardin{ \Bigl. \normX{ \pnt_i - \pntA - \DotProd{\pnt_i - \pntA}{\vecA\,} \vecA } - r} = \cardin{ \sqrt{p_i(\pntA, \vecA, r )} - r},
\]
where $p_i(\pntA, \vecA, r)$ is a polynomial with linearization dimension $O(d^4)$ (as can be easily verified), for $i=1,\ldots, n$.  The linearization dimension in this case can be reduced with more care, see \cite{ahv-aemp-04}. Thus, the overall price of fitting $\Cyl$ to the points of $P$ is $\sum_i f_i( \Cyl)$. This falls into our framework, and we get:
\begin{theorem}
    \thmlab{cylinder_fitting}%
    Let $P$ be a set of $n$ points in $\Re^d$, and $\eps >0$ a parameter. One can $(1+\eps)$-approximate the cylinder that best fits the points of $P$, where the price is the sum of Euclidean distances of the points of $P$ to the cylinder. The running time of the algorithm is $O(n + \poly( \log n, 1/\eps ))$, and the algorithm succeeds with high probability.

    Similarly, one can $(1+\eps)$-approximate the cylinder that minimizes the sum of square distances of the points of $P$ to the cylinder.
\end{theorem}

Interestingly, in two dimensions, an algorithm similar to the one in \thmref{cylinder_fitting} solves the problem of finding two parallel lines that minimize the sum of distances of the points to the lines (i.e., each point contributes its distance to the closer of the two lines).

\subsubsection{\TPDF{$1$}{1}-Median Clustering of Partial Data}

Consider an input of $n$ points $p_1,\ldots, p_n$, where the points are not explicitly given. Instead, we are provided with a set $\FlatSet = \brc{\Flat_1,\ldots, \Flat_n}$ of $n$ flats, such that $p_i \in\Flat_i$, where a \emph{flat} is an affine subspace of $\Re^d$. This naturally arises when we have partial information about a point and the point must comply with certain linear constraints that define its flat.

It is now natural to want to best-fit or cluster the partial data. For example, we might wish to compute the smallest ball that encloses all the partial points. This boils down to computing the smallest ball $\ball$ that intersects all the flats (i.e., we assume the real point $p_i$ lies somewhere in the intersection of the ball $\ball$ and $\Flat_i$).  An approximation algorithm for this problem that has polynomial dependency on the dimension $d$ (but bad dependency on the dimensions of the flats) was recently published by Gao \etal{} \cite{gls-aidid-06}.

Here, we are interested in finding the point $\Center$ that minimizes the sum of distances of the point $\Center$ to the flats $f_1,\ldots, f_n$. Namely, this is the $1$-median clustering problem for partial data.

Consider a flat $\Flat$ which contains the point $\pntA$, and is spanned by the unit vectors $\vecA_1,\ldots, \vecA_k$. That is $\Flat = \Set{ \pntA + t_1 \vecA_1 + \cdots + t_k\vecA_k}{ t_1, \ldots, t_k \in \Re}$. Then, we have that the distance of $\pnt \in \Re^d$ from the flat $\Flat$ is
\[
    \dmY{\pnt}{\Flat}%
    =%
    \normX{ \pnt - \pntA - \sum_{i=1}^k \DotProd{ \pnt - \pntA}{\vecA_i}\vecA_i} = \sqrt{\Polynomial(\pnt)},
\]
where $\Polynomial(\cdot)$ is a polynomial with linearization dimension $O(d^2)$. Thus, the problem of $1$-median clustering of partial data is no more than finding the point $\pnt$ that minimizes the function $\MdPrcX{\FlatSet}(\pnt) = \sum_i \dmY{\pnt}{ \Flat_i}$.  Approximating the minimum of this function can be done using \thmref{main}. We conclude:

\begin{theorem}
    \thmlab{clustering_partial}%
    Let $\FlatSet = \brc{\Flat_1, \ldots, \Flat_n}$ be a set of $n$ flats in $\Re^d$, and $\eps >0$ a parameter. One can compute a point $\pnt \in \Re^d$, such that $\MdPrcX{\FlatSet}(\pnt)$ is a $(1+\eps)$-approximation to $\min_{\pntA} \MdPrcX{\FlatSet}(\pntA)$.  The running time of the algorithm is $O(n + \poly(\log n, 1/\eps))$, and the algorithm succeeds with high probability.
\end{theorem}

Note that $1$-mean clustering in this case is trivial as it boils down to a minimization of a quadratic polynomial.

\section{Conclusions}
\seclab{conclusions}

In this paper, we have described in this paper a general approximation technique for problems of $L_1$-fitting of a shape to a set of points in low dimensions.  The running time of the new algorithm is $O( n + \poly( \log n, 1/\eps))$, which is a linear running time for a fixed $\eps$.  The constant powers hiding in the polylogarithmic term are too embarrassing to be explicitly stated. Still, they are probably somewhere between $20$ to $60$, even just for the problem of $L_1$-fitting a circle to a set of points in the plane. Namely, this algorithm is only of theoretical interest. As such, the first open problem raised by this work is to improve these constants. A considerably more interesting problem is to develop a practical algorithm for this family of problems.

A natural, tempting question is whether one can use the techniques in this paper for the problem of $L_1$-fitting a spline or a Bezier curve to a set of points. Unfortunately, the resulting surfaces in the parametric space are no longer nice functions. Therefore, the algorithmic difficulty here is overshadowed by algebraic considerations. We leave this as an open problem for further research.

Another natural question is whether one can use the techniques of Har-Peled and Wang \cite{hw-sfo-04} directly, to compute a coreset for this problem, and solve the problem on the coreset directly (our solution did a similar thing, by breaking the parametric space into a small number of prisms, and constructing a small ``sketch'' inside each such region). This would be potentially a considerable simplification over our current involved and messy approach. There is unfortunately a nasty technicality that requires that a coreset for the $L_1$-fitting of a linear function is also a coreset if we take the square root of the functions (as holds for the construction of \secref{chunking}). It seems doubtful that this claim holds in general, but maybe a more careful construction of a coreset for the case of planes in three dimensions would still work. We leave this as an open problem for further research.

The author believes that the algorithm presented in this paper should have other applications. We leave this as an open problem for further research.

\paragraph*{Acknowledgments.}

The author thanks Pankaj Agarwal, Arash Farzan, John Fischer, Vladlen Koltun, Bardia Sadri, Kasturi Varadarajan, and Yusu Wang for helpful and insightful discussions related to the problems studied in this paper.

The author also thanks the anonymous referees for their detailed and patient comments on the manuscript.

\begingroup%
\emergencystretch=1em   %
\printbibliography[heading=subbibliography,segment=\therefsegment]{}%
\endgroup %

\end{document}